\documentclass[10pt, a4paper,reqno]{amsart}
\usepackage[margin=1.25in]{geometry}
\usepackage{multirow}
\usepackage{amssymb}
\usepackage{amsmath,amsthm,cite,enumerate,mathtools}
\usepackage{wasysym}
\usepackage{enumitem}
\usepackage{tensor}
\usepackage{booktabs}
\usepackage{xcolor}
\usepackage{amsaddr}

\usepackage{etoolbox}
\patchcmd{\subsection}{\textbf}{}{}{}
\patchcmd{\subsection}{-.5em}{.2em}{}{}
\makeatletter
\patchcmd{\@startsection}
  {\@afterindenttrue}
  {\@afterindentfalse}
  {}{}
\makeatother

\usepackage[utf8]{inputenc}

\usepackage{caption}
\newtheorem{theorem}{Theorem}[section]
\newtheorem{lemma}[theorem]{Lemma}

\theoremstyle{definition}

\newtheorem{remark}[theorem]{Remark}

\numberwithin{equation}{section}

\newcommand*\tho{\text{\thorn}}
\newcommand*\dho{\text{\dh}}
\DeclareMathOperator{\Tr}{Tr}

\newcommand*\Bell{\ensuremath{\boldsymbol\ell}}
\newcommand*\Bn{\ensuremath{\boldsymbol{n}}}
\newcommand*\Bm{\ensuremath{\boldsymbol m}}

\newcommand*\Bg{\ensuremath{\boldsymbol{g}}}
\newcommand*\BG{\ensuremath{\boldsymbol{G}}}
\newcommand*\BR{\ensuremath{\boldsymbol{R}}}
\newcommand*\BC{\ensuremath{\boldsymbol{C}}}
\newcommand*\BF{\ensuremath{\boldsymbol{F}}}
\newcommand*\BT{\ensuremath{\boldsymbol{T}}}

\newcommand*\BH{\ensuremath{\boldsymbol{H}}}
\newcommand*\BS{\ensuremath{\boldsymbol{S}}}
\newcommand*\BA{\ensuremath{\boldsymbol{A}}}

\newcommand*\de{\ensuremath{\textnormal{d}}}

\newcommand*\phip{\ensuremath{f}}

\newcommand{\pp}{{\it pp\,}-}

\def \T {\bigtriangleup  }

\def \bF {\mbox{\boldmath{$F$}}}
\newcommand{\be}{\begin{equation}}
\newcommand{\ee}{\end{equation}}
\newcommand{\beqn}{\begin{eqnarray}}
\newcommand{\eeqn}{\end{eqnarray}}

\newcommand{\ba}{\begin{array}}
\newcommand{\ea}{\end{array}}

\def \bl {\mbox{\boldmath{$\ell$}}}
\def \bn {\mbox{\boldmath{$n$}}}

\def \bF {\mbox{\boldmath{$F$}}}

\def\d{\mathrm{d}}
\def \bff {\mbox{\boldmath{$f$}}}
\def \T {\bigtriangleup}

\begin{document}
%----------------------------------------------------------------------------------------------------------------------------------

\title{Einstein-Maxwell fields with vanishing higher-order corrections}


\author{\footnotesize Martin Kuchynka$^\dagger$ and Marcello Ortaggio$^\star$}
\address{$^{\dagger,\star}$Institute of Mathematics of the Czech Academy of Sciences \\ \v Zitn\' a 25, 115 67 Prague 1, Czech Republic
\\
$^\dagger$Institute of Theoretical Physics, Faculty of Mathematics and Physics, Charles University in Prague,
V Hole\v{s}ovi\v{c}k\'ach 2, 180 00 Prague 8, Czech Republic}
\email{$^\dagger$ kuchynkm(at)gmail(dot)com}
\email{$^\star$ ortaggio(at)math(dot)cas(dot)cz}

%\author{}

%----------------------------------------------------------------------------------------------------------------------------------
%General info
%----------------------------------------------------------------------------------------------------------------------------------
%\subjclass[2000]{Primary 54C40, 14E20; Secondary 46E25, 20C20}
%\date{January 1, 2001 and, in revised form, June 22, 2001.}

%\keywords{}

%\maketitle

%----------------------------------------------------------------------------------------------------------------------------------
\begin{abstract}

We obtain a full characterization of Einstein-Maxwell $p$-form solutions $(\Bg,\BF)$ in $D$ dimensions for which all higher-order corrections vanish identically. These thus simultaneously solve a large class of Lagrangian theories including both modified gravities and (possibly non-minimally coupled) modified electrodynamics.
Specifically, both $\Bg$ and $\BF$ are fields with vanishing scalar invariants and further satisfy two simple tensorial conditions. They describe a family of gravitational and electromagnetic plane-fronted waves of the Kundt class and of Weyl type III (or more special). The local form of $(\Bg,\BF)$ and a few examples are also provided.

\end{abstract}
%----------------------------------------------------------------------------------------------------------------------------------

\date{\today}

\maketitle

%----------------------------------------------------------------------------------------------------------------------------------
\section{Introduction and summary}
\label{intro}

While the Einstein-Maxwell Lagrangian is generally considered to describe the prototype theory of gravity coupled to electromagnetism, there is also a long history of so-called ``alternative theories''. The long-standing problem of the electron's self-energy led to a modified electrodynamics already in 1912 \cite{Mie12} and subsequently to the well-known non-linear theory of Born and Infeld \cite{Born33,BorInf34} (see, e.g., \cite{Plebanski70} for more general non-linear electrodynamics (NLE)). Soon after the birth of General Relativity, the quest for a unified description of gravity and electromagnetism also inspired several modifications of Einstein's theory -- see, e.g., the early works \cite{Weyl18,Weyl19} and the reviews \cite{Goenner04,Goenner14} for more references. In subsequent years, further motivation to take into account deviations from the Einstein-Maxwell theory came from considering effective Lagrangians which include various type of quantum corrections (cf., e.g., \cite{Deser75,Dunne05} and the original references quoted there) or low-energy limits of string theory \cite{SchSch74,Callanetal85,FraTse85,AndTse88}.

Not surprisingly, adding higher-order corrections to the Einstein and Maxwell equations makes those generically more difficult to solve. However, it is remarkable that there exist theory-independent solutions, i.e., solutions ``immune'' to (virtually) any type of corrections. One can thus employ known solutions of the Einstein-Maxwell equations to explore more complicated theories, at least in certain regimes. This was first pointed out in the context of NLE by Schr\"odinger, who showed that all null fields which solve Maxwell's theory also automatically solve any NLE in vacuum \cite{Schroedinger35,Schroedinger43}. The inclusion of backreaction on the spacetime geometry in the full Einstein-Maxwell theory was later discussed in \cite{Peres61}. Subsequently, it was noticed that electromagnetic plane waves solve not only NLE but also higher-order theories \cite{Deser75} (in flat spacetime; see also \cite{Schwinger51}), and that a similar property is shared by Yang-Mills and gravitational plane waves \cite{Deser75}. Backreaction was taken into account in \cite{Guven87}, whereas extensions of these results to more general (electro)vacuum $pp$- and AdS-waves were obtained in \cite{AmaKli89,HorSte90,Horowitz90} and \cite{HorItz99}, respectively. This was used, in particular, to discuss spacetime singularities in string theory \cite{HorSte90,Horowitz90}.

Recently, a more systematic analysis of $D$-dimensional Einstein spacetimes immune to purely gravitational corrections (``universal spacetimes'') was initiated in \cite{Coleyetal08} and further developed in \cite{HerPraPra14,Herviketal15,HerPraPra17} (see also \cite{Gursesetal13,GurSisTek14,GurSisTek17} for related results in the case of Kundt (AdS-)Kerr-Schild metrics). From a complementary viewpoint, a study of test Maxwell fields which simultaneously solve also generalized theories of ($p$-form) electrodynamics (``universal electromagnetic fields'') has been performed in \cite{OrtPra16,OrtPra18,HerOrtPra18}. In spite of considerable progress, a full characterization of (i.e., a necessary and sufficient condition for) universal spacetimes and universal electromagnetic fields is, in general, still lacking (but see the above references for various results in special cases).

In the present contribution we investigate solutions of the {\em coupled} (possibly also non-minimally) Einstein-Maxwell equations for which all higher-order corrections vanish identically in arbitrary dimension $D$ and for any rank $p$ of the Maxwell form. We show that a full characterization is possible, which we formulate as theorems~\ref{minuniversal} and \ref {nonminuniversal}. Essentially (up to technicalities to be explained in the following), we prove that for a solution $(\Bg,\BF)$ of the Einstein-Maxwell theory, {\em all higher-order corrections vanish if, and only if, both $(\Bg,\BF)$ are fields with vanishing scalar invariants ($VSI$) and additionally satisfy the two tensorial conditions $C_{acde}C\indices{_{b}^{cde}} = 0$ and $\nabla_c F_{ad \dots e} \nabla^c F\indices{_b^{d \dots e}} = 0$}. This implies, in particular, that the spacetime is Kundt and possesses a recurrent null vector field (but is not necessarily a \pp wave) and that the cosmological constant vanishes. This characterization of a large class of exact solutions make those relevant in contexts more general than the Einstein-Maxwell theory, with possible applications, e.g., in string theory along the lines of \cite{Guven87,AmaKli89,HorSte90,Horowitz90}. Moreover, the methods used in this work are suitable also for further extensions of the results obtained here, for example to Yang-Mills solutions. 

The structure of the paper is as follows. In section~\ref{theories} we define the theories under considerations and in what sense those can be considered as corrections to the Einstein-Maxwell theory. Section~\ref{universalsolutions} contains the main results of this paper, namely theorems~\ref{minuniversal} and \ref {nonminuniversal} (in the case of minimally and non-minimally coupled theories, respectively) and their proofs. A simpler result for the special case of Einstein gravity coupled to {\em algebraically} corrected electrodynamics (relevant for theories similar to NLE) is also obtained (theorem~\ref{algebraic}). In section~\ref{typeIIIsolutions}, we present the explicit form of the solutions $(\Bg,\BF)$ in adapted coordinates, which is more suitable for practical applications, along with a few examples. The relation of the solutions to universal spacetimes \cite{Coleyetal08,HerPraPra14,Herviketal15,HerPraPra17} and universal electromagnetic fields \cite{OrtPra16,OrtPra18,HerOrtPra18} is also discussed, along with the overlap with Kerr-Schild spacetimes. Some additional comments are provided in the special case of four spacetime dimensions. The four appendices contain various technical results used throughout the paper (in particular, in the proofs of the main theorems). Most of those are new and of some interest in their own, and we believe they will be useful also in future investigations (we have quoted the relevant references in the few cases in which we simply summarize previously known results).

\subsection*{Notation}
Throughout the paper, we employ the boost-weight classification of tensors \cite{Milsonetal05} (cf. also the review \cite{OrtPraPra13rev}) -- this relies on setting up a
frame of $D$ real vectors $\Bm_{(a)} $ which consists of two null vectors $\bl\equiv{\mbox{\boldmath{$m_{(0)}$}}}$,  $\bn\equiv{\mbox{\boldmath{$m_{(1)}$}}}$ and $D-2$ orthonormal spacelike vectors $\Bm_{(i)} $ (with $a, b\ldots=0,\ldots,D-1$ and $i, j, \ldots=2,\ldots,D-1$), such that the metric reads
\be
	\Bg=\bl\otimes\bn+\bn\otimes\bl+\Bm_{(i)}\otimes\Bm_{(i)} .
	\label{g_null}
\ee	
The range of lowercase Latin indices when indicating an order of differentiation (e.g., in $\nabla^{(k)} \BR$) will be specified as needed.
Furthermore, $\BR$, $\BC$, $\BS$ denote the Riemann and Weyl tensors and the tracefree part of the Ricci tensor (cf. \eqref{S}), respectively. A $p$-form is denoted by $\BF$. 
A ``Maxwell $p$-form'' is a $p$-form which obeys the sourcefree Maxwell equations, i.e., $\d\bF=0=\d\star\bF$.

%----------------------------------------------------------------------------------------------------------------------------------
\section{Higher order theories of gravity and electromagnetism}
\label{theories}
%----------------------------------------------------------------------------------------------------------------------------------

\subsection{Form of the Lagrangian}

\label{subsec_lagrangian}

In the paper, we take into account virtually all classical Lagrangian theories of gravity coupled to electromagnetism, described by the electrovacuum Einstein-Maxwell equations with higher-order corrections. More precisely, we
consider a theory of gravity and $p$-form electromagnetism, in spacetime dimensions $D\ge3$ and with $1\le p\le D-1$,\footnote{As well-known, a Maxwell $D$-form reduces to the spacetime volume element (up to a multiplicative constant) and simply gives rise to an effective positive cosmological constant, so that a spacetime with vanishing higher-order corrections must be Einstein (for $D=2$ this simply fixes the value of $\Lambda$ in terms of $\BF$). The cases $p=D$ and, by duality, $p=0$, are thus of little interest in our work. We also exclude the case $D=2$ with $p=1$, since Einstein's equations imply the trivial condition $\BF=0$. This is why we restrict ourselves to $D\ge3$.} characterized by the action 
\begin{equation}\label{action}
S[\Bg,\BA] =  \int \textnormal{d}^D x \sqrt{-g} \mathcal{L},
\end{equation}
with a Lagrangian $\mathcal{L}$ of the form
\begin{equation}\label{lagrangian}
\mathcal{L} \equiv \mathcal{L}_{grav}(\BR,\nabla \BR, \dots)+ \mathcal{L}_{elmag}(\BF, \nabla \BF, \dots) + \mathcal{L}_{int}(\BR,\nabla \BR, \dots, \BF, \nabla \BF, \dots) .
\end{equation}
Here, the individual parts of $\mathcal{L}$ are scalars constructed from the corresponding tensors: $\BR$ denotes the Riemann tensor of the metric $\Bg$, and $\BF$ denotes the field strength of the electromagnetic potential $(p-1)$-form $\BA$, i.e. $\BF = \textnormal{d}\BA$. 
We assume that the individual parts of $\mathcal{L}$ satisfy:
\begin{itemize}
\item  $\mathcal{L}_{grav}$ is a function of scalar polynomial curvature invariants $\{I_i\}$ constructed from $\BR$ and its covariant derivatives $\nabla^{(k)}\BR$ of arbitrary order (suitably contracted with the the metric and, possibly, the volume element).  Moreover, $\mathcal{L}_{grav}(I_1, I_2, \dots)$ is analytic at zero with a Taylor expansion of the form 
\begin{equation}
 \mathcal{L}_{grav} = \mathcal{L}_{EH} + \mathcal{L}_{GC},
 \label{exp_grav}
\end{equation}
where 
\be
 16\pi\mathcal{L}_{EH}=R-2\Lambda 
\ee 
defines the Einstein-Hilbert Lagrangian (we have set $G=1=c$), and $\mathcal{L}_{GC}$ (``Gravity Corrections'') consists strictly of higher order (i.e., greater than two) curvature monomials.\footnote{Following the terminology of \cite{Fullingetal92}, throughout the paper by ``order'' we indicate the number of differentiations of the metric/vector potential (so, for example, in terms containing the curvature, each factor $\BR$ contributes a term 2 and each explicit covariant derivative a term 1 \cite{Fullingetal92}). Two quantities of the same order have thus the same physical dimensions. Most importantly, the field variation of an invariant of order $n$ (in our case, w.r.t. $\Bg$ or $\BA$) yields a tensor again of {\em the same} order $n$.} This means that the possible monomials are at least quadratic in $\BR$ or contain derivatives $\nabla^{(k)} \BR$.

\item  $\mathcal{L}_{elmag}$ is a function of scalar polynomial electromagnetic invariants $\{J_j\}$ constructed from $\BF$ and $\nabla^{(k)}\BF$ of arbitrary order. Moreover, 
 $\mathcal{L}_{elmag}(J_1, J_2, \dots)$ is analytic at zero with a Taylor expansion of the form
\begin{equation}
 \mathcal{L}_{elmag} = \mathcal{L}_{M} + \mathcal{L}_{EC},
\label{exp_Maxw}
\end{equation}
where 
\be
 16\pi\mathcal{L}_{M}=-\frac{\kappa_0}{p}F^2   \qquad (F^2=F_{ab\ldots c}F^{ab\ldots c}) ,
 \label{L_M}
\ee 
defines the source-free Maxwell Lagrangian, and $\mathcal{L}_{EC}$ (``Electromagnetic Corrections'') consists strictly of higher order  (i.e., greater than two) monomials. 

\item  $\mathcal{L}_{int}$ is a function of mixed invariants $\{K_k\}$ (i.e., scalar monomials each containing both $\BR,\nabla \BR, \dots$ and $\BF,\nabla \BF, \dots$) and satisfies $\mathcal{L}_{int}(0) = 0$.  
\end{itemize}

The above assumptions ensure that when the invariants entering $\mathcal{L}$ are small, $\mathcal{L}$ approaches the standard Einstein-Maxwell $p$-form Lagrangian, i.e., $16\pi\mathcal{L}\approx R-2\Lambda -\frac{\kappa_0}{p}F^2$. However, $\mathcal{L}$ is not assumed to be analytic everywhere -- as is the case for some of the theories mentioned in remark~\ref{rem_theories} below.

\begin{remark}[Theories contained in our definition]
\label{rem_theories}
The class of theories encompassed by \eqref{action}, \eqref{lagrangian} (with \eqref{exp_grav}--\eqref{L_M}) is rather broad. It naturally includes Einstein's gravity coupled to NLE \cite{Peres61} for arbitrary $D$ and $p$ (see section~\ref{subsubsec_algebr} below). Obviously, it also contains theories with arbitrary polynomial higher-order corrections, such as generic Lovelock \cite{Lovelock71} or any quadratic gravity \cite{Weyl19,Eddington_book,Lanczos38,Buchdahl48}  in the gravitational sector, or Bopp-Podolsky electrodynamics \cite{Bopp40,Podolsky42} in the electromagnetic sector. Nonlinear theories such as $f(R)$ \cite{Buchdahl70} and, more generally, $f(\textnormal{Riemann})$ \cite{Deruelleetal10}, or Born-Infeld inspired modifications of gravities \cite{DesGib98} coupled to generalized electrodynamics (such as NLE and their various generalizations) are also encompassed. Another special class of theories covered by \eqref{action} are then non-minimally extended Einstein-Maxwell theories (see, e.g., \cite{Prasanna71} for an early discussion).

Also some theories {\em not} encompassed by our assumptions are worth mentioning. These are typically theories without the Einstein term in the gravity sector, such as conformal gravity \cite{Bach21} or any Lovelock gravity containing only quadratic or higher powers of $\BR$ (e.g., pure Gauss-Bonnet gravity). We observe that also theories containing an electromagnetic Chern-Simons (CS) term (possible for $D=p(k+1)-1$, where $k\ge1$ -- cf., e.g., \cite{CreJulSch78,Banadosetal97}) are not comprised in our definition. However, since any null $\bF$ 
satisfies identically $\bF\wedge\bF=0$, CS corrections to the Maxwell equations with $k\ge2$ vanishes identically for null fields \cite{FigPap01,OrtPra16}. The energy-momentum tensor is also unaffected, therefore the solution of theorems~\ref{minuniversal} and \ref {nonminuniversal} are also immune to CS corrections.  In the special case $k=1$, CS corrections to the Maxwell equations are instead linear and therefore a non-zero solution of standard Maxwell's theory cannot solve those.
 \end{remark}

\subsection{Field equations}
Variation of action \eqref{action} with respect to the fields $\Bg$ and $\BA$ yields the following equations of motion
\begin{equation}\label{EM1}
G^{grav}_{ab} + G^{int}_{ab} =  8\pi T^{elmag}_{ab},
\end{equation}
\begin{equation}\label{EM2}
\nabla^{a} H^{elmag}_{a b \dots c} + \nabla^{a} H^{int}_{a b \dots c} = 0.
\end{equation}
From the Taylor expansion \eqref{exp_grav} and \eqref{exp_Maxw} of $\mathcal{L}_{grav}$ and $\mathcal{L}_{elmag}$, respectively, the following expressions for the individual tensors in \eqref{EM1}, \eqref{EM2} follow (within the radii of convergence of the Taylor series): 
\begin{align}
&G^{grav}_{ab} = G_{ab} + \Lambda g_{ab}  +  G^{GC}_{ab},& 
&G^{GC}_{ab} \equiv \frac{16\pi}{\sqrt{-g}}\frac{\delta(\sqrt{-g}\mathcal{L}_{GC})}{\delta g^{ab}}, \\
&T^{elmag}_{ab} = T^{M}_{ab} + T^{EC}_{ab},& 
&T^{EC}_{ab} \equiv \frac{-2}{\sqrt{-g}}\frac{\delta(\sqrt{-g}\mathcal{L}_{EC})}{\delta g^{ab}}, \\
&\nabla^a H^{elmag}_{a b \dots c} =  \nabla^a F_{a b \dots c} + \nabla^a H^{EC}_{a b \dots c},& 
&\nabla^a H^{EC}_{a b \dots c} \equiv \frac{8 \pi}{\kappa_0}\frac{\delta L_{EC}}{\delta A^{ b \dots c}},
\end{align}
where
\be
 G_{ab}\equiv\frac{16\pi}{\sqrt{-g}}\frac{\delta(\sqrt{-g}\mathcal{L}_{EH})}{\delta g^{ab}}=R_{ab}-\frac{1}{2}Rg_{ab} , 
\ee 
\be
 T^{M}_{ab} \equiv \frac{-2}{\sqrt{-g}}\frac{\delta(\sqrt{-g}\mathcal{L}_{M})}{\delta g^{ab}}=\frac{\kappa_{0}}{8\pi}\left(F_{a c \dots d}F\indices{_{b}^{c \dots d}} - \frac{1}{2p}g_{ab} F^2\right) , 
 \label{TM}
\ee 
are the Einstein tensor and the part of the energy-momentum tensor coming from the standard Maxwell term.
The interaction tensors $\BG^{int}$, $\textnormal{div}\BH^{int}$ are then a symmetric and skew-symmetric tensor obtained by the field variation of $\mathcal{L}_{int}$ with respect to $\Bg$ and $\BA$, respectively. The explicit form of variations of $ \mathcal{L}_{GC},  \mathcal{L}_{EC}$ and $\mathcal{L}_{int}$ evaluated on $VSI$ fields (which suffices for our purposes) is given in appendix~\ref{app_variations} (expressions \eqref{var1}--\eqref{var5}).

\begin{remark}[Simplifications of CSI and VSI fields] \label{CSIvariations}
When evaluated on fields $(\Bg,\BF)$ with constant scalar invariants ($CSI$), variations of Lagrangians $ \mathcal{L}_{grav},  \mathcal{L}_{elmag}$ and $\mathcal{L}_{int}$ being functions of the corresponding scalar polynomial invariants $\{I_i\},\{J_j\}$ and $\{K_k\}$, respectively, reduce to a linear combination (with constant coefficients) of variations of these scalar invariants (see appendix \ref{varapp}). This means that the fields equations reduce considerably for such fields -- in particular, it enables one to study (in general complicated) higher-order theories in the context of $CSI$ fields just by studying field variations of the individual scalar polynomial invariants, independently of the specific functional dependence of the Lagrangian. Further simplification occurs in the case of VSI fields (clearly a subset of CSI fields). In the next section, this strategy will be employed in the proofs of the main results. 
\end{remark}

%----------------------------------------------------------------------------------------------------------------------------------
\section{Solutions with vanishing higher-order corrections}
\label{universalsolutions}
%----------------------------------------------------------------------------------------------------------------------------------
We will show that under certain assumptions on a solution $(\Bg,\BF)$ of the Einstein-Maxwell equations, the tensors $\BG^{GC},  \BG^{int}, \BT^{EC}, \textnormal{div} \BH^{EC}$ and $\textnormal{div} \BH^{int}$, representing higher-order corrections to the Einstein-Maxwell theory, vanish identically. 
Minimally coupled ($\mathcal{L}_{int} = 0$) and non-minimally coupled ($\mathcal{L}_{int} \neq 0$) theories will be treated separately.   

\subsection{Minimally coupled theories}\label{mincoupled}
In the minimally coupled case one has
\be
 \mathcal{L}_{int} = 0 ,
 \label{minim}
\ee 
so that the interaction tensors $\BG_{int}$ and $\BH_{int}$ are not present in the field equations \eqref{EM1}, \eqref{EM2}. 
Consequently,  we shall deal with simpler higher-order corrections to the Einstein-Maxwell system.

\begin{theorem}[Solutions with vanishing corrections]
\label{minuniversal}
Let $(\Bg,\BF)$ be a solution of the Einstein-Maxwell theory with a non-vanishing $\BF$ and \eqref{action}, \eqref{lagrangian} be a minimally coupled theory (i.e. with $\mathcal{L}_{int}=0$) satisfying the assumptions outlined in section~\ref{subsec_lagrangian}. Then, the following statements are equivalent: 
\begin{enumerate}
\item \label{x} All higher-order corrections of \eqref{action} to the Einstein-Maxwell theory vanish for $(\Bg,\BF)$. 
\item $(\Bg,\BF)$ are $VSI$ fields and satisfy $C_{acde}C\indices{_{b}^{cde}} = 0$ and $\nabla_c F_{ad \dots e} \nabla^c F\indices{_b^{d \dots e}} = 0$. 
\end{enumerate}
\end{theorem}

\begin{remark}
\label{rem_ii}
First of all, let us note that the VSI property in condition~(ii) of theorem~\ref{minuniversal} requires the cosmological constant $\Lambda$ to be zero. Condition~(ii) also implies that the spacetime is of Weyl type III \cite{Coleyetal04vsi} and admits a {\em recurrent} multiple Weyl aligned null direction (mWAND) $\bl$ {\em aligned} with $\bF$ (see remark~\ref{rem_DFDF_2}), thus being Kundt. Note also that the condition $C_{acde}C\indices{_{b}^{cde}} = 0$ can be traced back to the vanishing of the Gauss-Bonnet term in the gravitational field equations (as such, it has been discussed in related contexts, e.g., in \cite{PraPra08,MalPra11prd,HerPraPra14,Ortaggio18prd}). We further emphasize that it is satisfied identically by VSI spacetimes in $D=4$ dimensions, thanks to the well-known four-dimensional identity $C_{acde}C^{bcde}=\frac{1}{4}(C_{cdef}C^{cdef})\delta_a^b$. For $D=3$ it is also trivial since $C_{abcd}=0$ identically.
\end{remark}

\begin{proof}
Let us first show that $(i)$ implies $(ii)$. 
Consider the $2N$-th order Lagrangian $\mathcal{L}_{EC} \equiv J_1^N$ ($N>1$), where $J_1 = F_{a \dots b} F^{a \dots b}$. The condition $\BT^{EC} = 0$ implies that the trace $\Tr \BT^{EC} = -2 (Np - D/2) J_1^N$ has to vanish and hence necessarily $J_1 = 0$, since $N$ can be chosen arbitrarily.  
Now, one can take $\mathcal{L}_{EC} \equiv J_1 I$, where $I$ is an arbitrary scalar polynomial invariant of $\BF$ and its covariant derivatives. Thanks to $J_1 = 0$, the corresponding correction reduces to $\BT^{EC} \propto I \BT^{M}$ for our field $\BF$, where $\BT^M$ is the standard Maxwell energy-momentum tensor \eqref{TM} (which is necessarily non-zero since $\BF\neq0$). Hence, also $I$ has to vanish and, since it was an arbitrary invariant, $\BF$ is $VSI$. In particular, it is null and the metric $\Bg$ is (degenerate) Kundt of traceless Ricci type N with constant Ricci scalar \cite{OrtPra16}. The condition $\BG^{GC} = 0$ then implies that also $\Bg$ has to be $VSI$. 
Indeed, considering $\mathcal{L}_{GC} = R^2$, we get $R=0$.\footnote{We do not reproduce here the tensors produced by variation of such kinds of Lagrangians w.r.t. the metric since they have been well-known for some time \cite{Buchdahl48,DeWittbook}. The same comment applies also to the other quadratic terms mentioned in the following.}
 This suffices to conclude that $\Bg$ is $CSI$, as immediately follows from (the proof of) theorem~3.2 of \cite{HerPraPra14} (using $\Tr\BG^{GC} = 0$). 
Then, varying $\mathcal{L}_{GC} = R I$ with $I$ being an arbitrary scalar polynomial curvature invariant (also using $R=0$ and the $CSI$ property of $\Bg$), one obtains that $I$ has to vanish as well, i.e. $\Bg$ is truly $VSI$. In particular, $\Bg$ is  of aligned Weyl type III and Ricci type N \cite{Coleyetal04vsi} (in addition to being degenerate Kundt). 
In view of the results obtained so far, varying the higher-order invariants $R_{ab}R^{ab}$ and $R_{abcd}R^{abcd}$ and demanding that such corrections also vanish, we obtain that $\Box S_{ab}$, and consequently also $C_{acde}C\indices{_{b}^{cde}}$, vanishes. 
Under the given conditions on $(\Bg,\BF)$, the Weitzenb\"ock identity implies $\Box \BF = 0$ (cf. eq.~(12) of \cite{HerOrtPra18}). Since here $S_{ab}=\kappa_0F_{a c \dots d}F\indices{_{b}^{c \dots d}}$ (by Einstein's equations with null $\BF$), we have that $\Box S_{ab}=0$ iff $\nabla_c F_{ad \dots e} \nabla^c F\indices{_b^{d \dots e}} = 0$, which completes the first part of the proof.

Now we will prove that $(ii)$ implies $(i)$. First, both fields are $VSI$, thus,  as pointed out in remark \ref{CSIvariations}, all higher-order corrections of \eqref{action} reduce to a linear combination of variations of the individual polynomial invariants $I_k,J_k,K_i$ (see expressions \eqref{var1}--\eqref{var5} and the discussion below those). Hence, the discussion can be without loss of generality restricted to polynomial higher-order corrections $\BG^{GC}, \BT^{EC}$ and $\BH^{EC}$. 
Now, according to theorem~1 of \cite{Coleyetal04vsi}, $\Bg$ is of aligned Weyl type III and Ricci type N,  and thus also aligned with the $VSI$ form $\BF$ (thanks to Einstein's equations). Theorem 2.5 of \cite{HerOrtPra18} then implies $\textnormal{div} \BH^{EC} = 0$.  
In view of theorem \ref{1balcriteria} (with remark~\ref{rem_DFDF_2}),  $\nabla \BF$ is 1-balanced, all conditions of lemma \ref{rank2F} are satisfied and consequently $\BT^{EC} = 0$ (recall that $\BT^{EC}$ has order greater than two). It remains to show that $\BG^{GC}$ vanishes as well. Since $\BF$ is a null Maxwell field aligned with a Kundt null direction $\Bell$, from remark~\ref{rem_DFDF_2} we get $\tau_i=0$, i.e. $\Bell$ is recurrent. Theorem~\ref{Tform} thus guarantees that $\BG^{GC}$ takes the form 
\begin{equation}\label{formaG}
G^{GC}_{ab} = \sum_{n=0}^N a_n \Box^n S_{ab}.
\end{equation}
As noticed above, here $S_{ab}=\kappa_0F_{a c \dots d}F\indices{_{b}^{c \dots d}}$. Hence, $1$-balancedness of $\nabla \BF$ implies $\Box \BS = 0$ and we are left with $G^{GC}_{ab} = a_0 S_{ab}$. But since $\mathcal{L}_{GC}$ is a higher-order scalar, $a_0$ must be a non-trivial curvature invariant and hence vanishes due to the $VSI$ property of $\Bg$.
\end{proof}

\subsubsection{Algebraic corrections to the Maxwell Lagrangian}
\label{subsubsec_algebr}
A subclass of theories of particular interest consists of standard General Relativity coupled to generalized electrodynamics, for which the higher-order corrections are assumed to be only {\em algebraic}. This includes, in particular, the well-known case of NLE in four dimensions \cite{Peres61}. Let us thus consider Einstein-generalized Maxwell theories with algebraic corrections, i.e. a subclass of theories \eqref{action}, \eqref{lagrangian} for which the expansion~\eqref{exp_grav}, \eqref{exp_Maxw} reduces to
\begin{equation}
\mathcal{L} = \mathcal{L}_{EH} + \mathcal{L}_{M} + \mathcal{L}_{EC},
\end{equation} 
where  $\mathcal{L}_{EC}$ is a (higher-order) function of the \textit{algebraic} invariants $\{J_j\}$ only (i.e. those constructed solely from $\BF$ and its dual, and not their covariant derivatives).

\begin{theorem}[Einstein gravity with algebraically corrected electrodynamics]\label{algebraic}
Let $(\Bg,\BF)$ be a solution of the Einstein-Maxwell equations with non-vanishing $\BF$. Then, $(\Bg,\BF)$ solves Einstein gravity coupled to any generalized Maxwell theory with higher-order algebraic corrections if and only if $\BF$ is null. 
\end{theorem}
\begin{proof}
To prove that $F$ is necessarily null, we can proceed similarly as in the proof of theorem \ref{minuniversal}. By considering $\mathcal{L}_{EC} \equiv J_1^N$, where $J_1 \equiv F_{a \dots b} F^{a \dots b}$, one obtains $J_1 = 0$ for a suitable choice of $N$. Now, the Lagrangian $\mathcal{L}_{EC} \equiv J_1 I$, where $I$ is an arbitrary algebraic polynomial invariant of $\BF$, is clearly an admissible correction. Since $J_1$ vanishes for $\BF$, we have $\BT^{EC} \propto I \BT^M$, which implies that also $I = 0$. Therefore, all algebraic invariants of $\BF$ must vanish, i.e., $\BF$ is null \cite{Hervik11,OrtPra16}.

On the other hand, since all algebraic invariants $\{J_j\}$ of any null field $\BF$ vanish, the tensors $\BH^{EC}$ and $\BT^{EC}$ again effectively reduce to {\em polynomial} higher-order corrections (cf. appendix~\ref{app_variations}). Therefore,  we have $\textnormal{div} \BH^{EC}=0$ thanks to Proposition~2.4 of \cite{HerOrtPra18}. In addition, since any higher-order algebraic polynomial $\BT^{EC}$ has to be at least cubic in $\BF$, one also immediately obtains $\BT^{EC} = 0$, i.e. all algebraic higher-order corrections vanish trivially. 
\end{proof}

Hence, we observe that null Einstein-Maxwell fields are indeed of particular importance in the context of higher-order theories. It is worth emphasizing that, in this case, the metric is restricted neither to be of Weyl type III nor Kundt, and $\Lambda$ can be non-zero, thus allowing for more general spacetimes. Many such solutions are known in the case $D=4=2p$ (cf. \cite{Stephanibook} and references therein). In higher dimensions, some non-Kundt solutions have been presented, e.g., in \cite{OrtPodZof15} (when  $D=2p$). A simple Weyl type D example with $D=6=2p$ is given by \cite{OrtPodZof15}
\beqn
  & & \d s^2=r^{2}\delta_{ij}\d x^i\d x^j{+}2\d u\d r+\left(\frac{\Lambda}{10}r^2+\frac{\mu(u)}{r^3}\right)\d u^2  \qquad (i,j,\ldots=2,\ldots,5) \\
	& & \boldsymbol{F} = \frac{1}{2} f_{ij}(u) \de u \wedge \de x^{i} \wedge \de x^{j}  , \qquad \mu(u)=\mu_0{+}\frac{\kappa_0}{2}\int(f_{ij}f^{ij})\d u ,
\eeqn
where $\mu_0$ is a constant, which describes (for $\Lambda<0$) the formation of asymptotically locally AdS black holes by collapse of electromagnetic radiation with non-zero expansion.

Note also that, for the case $D=4=2p$, it was already known to Schr\"odinger that all null Maxwell fields automatically solve any NLE \cite{Schroedinger35,Schroedinger43}, while the fact that all null solutions of the Einstein-Maxwell theory solve also General Relativity coupled to any NLE was pointed out in the early 60's \cite{Peres61} (see also \cite{Kichenassamy59,KreKic60}).

%----------------------------------------------------------------------------------------------------------------------------------
\subsection{Non-minimally coupled theories}\label{sec_nonminuniversal}
%----------------------------------------------------------------------------------------------------------------------------------

In this section, we show that the Einstein-Maxwell solutions studied in section \ref{minuniversal} are free from higher order corrections also in the context of a wider class of {\em non-minimally coupled} theories -- that is, also the interaction part of the field equations \eqref{EM1}, \eqref{EM2} amounting to  $\mathcal{L}_{int}$ vanishes identically for these Einstein-Maxwell fields.

\begin{theorem}\label{nonminuniversal}
The Einstein-Maxwell fields  with vanishing higher-order corrections of section~\ref{mincoupled} solve also all non-minimally coupled theories 
\eqref{action}, \eqref{lagrangian}. 
\end{theorem}
\begin{proof}
It is sufficient to show that the tensors $\BG^{int}$,  $\textnormal{div} \BH^{int}$ arising from $\mathcal{L}_{int}$ vanish -- since the vanishing of the other terms in \eqref{EM1}, \eqref{EM2} clearly follows by the same arguments as in the proof of $(ii)\Rightarrow (i)$ in theorem~\ref{minuniversal}. Again, we can without loss of generality restrict ourselves to polynomial higher-order corrections. 
Clearly, when both $\BG^{int}$ and $\BH^{int}$ consist of monomials containing $\nabla^{(k)}\BR \ast \nabla^{(l)}\BF$ with $k\ge0,l>0$, then a trivial boost weight (b.w.) counting shows that they have to vanish (recall lemmas~\ref{lemma_1deriv} and \ref{lemma_VSI} and the fact that $\nabla\BF$ is 1-balanced). This argument does not apply to $\BG^{int}$ in the case $l=0$ -- which however is covered by lemma \ref{rank2FC}.

However, different forms of $\BG^{int}$ and $\BH^{int}$ are also possible. Namely, if $\mathcal{L}_{int} \propto \nabla^{(k)}\BR \ast \nabla^{(l)}\BF$,  variations with respect to $\Bg$ and $\BA$ may yield 
terms of type $\nabla^{(k+2)}\ast \nabla^{(l)}\BF$ and $\nabla^{(l+1)} \ast \nabla^{(k)} \BR$, respectively. Fortunately, even these two types of terms are safe -- the first one is zero by lemma~\ref{rank2F} and the second one vanishes thanks to lemma \ref{formlemma}.
Hence, we conclude that also interaction terms necessarily vanish for $(\Bg,\BF)$. 
\end{proof}

%----------------------------------------------------------------------------------------------------------------------------------
\section{Explicit form of the solutions and discussion}\label{typeIIIsolutions}
%----------------------------------------------------------------------------------------------------------------------------------

The local form of general $VSI$ fields $(\Bg,\BF)$ solving the Einstein-Maxwell equations is known in standard Kundt coordinates (see \cite{Coleyetal06} and \cite{OrtPra16}\footnote{There is a typo in (6,\cite{OrtPra16}): the factorial $p!$ should be simply $p$.}). For solutions with vanishing higher-order corrections (theorem~\ref{minuniversal}) the multiple null direction (of both $\BR$ and $\BF$) must be recurrent (remark~\ref{rem_ii}), which gives (in the metric~\eqref{g_null})
\begin{equation}\label{ppframeIII}
\Bell = \de u, \qquad \Bn = \de r + \left[H^{(1)}(u,x)r + H^{(0)}(u,x)\right] \de u + W_k (u,x) \de x^k, \qquad \Bm_{(i)} = \de x^i, 
\end{equation}
\begin{equation}\label{localformIII}
\boldsymbol{F} = \frac{1}{(p-1)!} f_{i \dots j}(u) \de u \wedge \de x^{i} \wedge \dots \wedge \de x^{j}  \qquad (i,j,k,\ldots=2,\ldots,D-1) ,
\end{equation}
where we have also used the condition $\nabla_c F_{ad \dots e} \nabla^c F\indices{_b^{d \dots e}} = 0$ in $(ii)$ of theorem~\ref{minuniversal} to constraint the form of $\BF$ (cf. remarks~\ref{rem_DFDF_2}, \ref{rem_DFDF_vsi}). The functions $H^{(0)}$, $H^{(1)}$, $W_{i}$ and $f_{i \dots j}$ are then subject to the following equations 
\begin{equation}\label{Weylcond}
W_{[i,j]k}W^{[i,j]k} = 
2 W\indices{_{[k,m]}^{m}}W\indices{^{[k,n]}_{n}},
\end{equation}
\begin{equation}\label{-1ricci}
H\indices{^{(1)}_{,j}} = W\indices{_{[j,k]}^{k}},
\end{equation}
\begin{equation}\label{-2ricci}
 \Delta H^{(0)} =  2 H\indices{^{(1)}_{,k}}W^{k} + H^{(1)}W\indices{^{k}_{,k}}  + W_{[m, n]} W^{[m , n]} + {W_{m,u}} \indices{^{m}}  
- \kappa_0 \mathcal{F}^2,
\end{equation}
where $\Delta$ is the Laplace operator in the (flat) transverse space and $\mathcal{F}^2 \equiv f_{i \dots j} f^{i \dots j}$ was defined. 
Equation~\eqref{Weylcond} is equivalent to the condition $C_{acde}C\indices{_{b}^{cde}} = 0$ ($(ii)$ of theorem~\ref{minuniversal}), while
equations \eqref{-1ricci} and \eqref{-2ricci} correspond to the Einstein equations of negative boost weight (cf. \cite{Coleyetal06,OrtPra16}). The rest of projections of the Einstein-Maxwell equations is already satisfied \cite{Coleyetal06,OrtPra16}. For the sake of definiteness, an explicit example with $D\ge6$ and $p=3$  (building on an example given in \cite{Ortaggio18prd}) is given by \eqref{ppframeIII} with
\begin{equation}
\boldsymbol{F} = \de u \wedge(f_{23}\de x_2\wedge\de x_3+f_{45}\de x_4\wedge\de x_5) ,
 \label{F_ex6D}
\end{equation}
\beqn
 & & W_2=ax_3^2 , \qquad H^{(1)}=ax_2 ,  \\
 & & H^{(0)}=\frac{a^2}{3}x_3^4+bx_3^2+cx_4^2 , \qquad b+c=-\kappa_0(f_{23}^2+f_{45}^2) , \label{H0_ex6D}
\eeqn
where $f_{23}$, $f_{45}$, $a$, $b$ and $c$ are arbitrary functions of $u$ and the remaining $W_{i}$ ($i>2$) are understood to be zero.

The above spacetimes are generically of Weyl type III \cite{Coleyetal06}. The {\em Weyl type N} subclass of solutions takes the form \eqref{ppframeIII}, \eqref{localformIII} with the constraints (after using some coordinate freedom) \cite{Schimming74}\footnote{The first of \eqref{WC3pp} was obtained in \cite{Coleyetal06} and means that these solutions belong to the class of $pp$-waves (i.e., $\bl$ is covariantly constant). Then, for \pp waves, the ``ebenfrontiger Symmetrie'' condition~(2.1) of \cite{Schimming74} is equivalent to imposing the Riemann type N, which allows one to use theorem 2.1 of \cite{Schimming74} to arrive at the second of \eqref{WC3pp} (cf. also \cite{Horowitz90}).}
\begin{equation}\label{WC3pp}
 H^{(1)} = 0 , \qquad  W_{i}= 0 .
\end{equation}
Eqs.~\eqref{Weylcond} and \eqref{-1ricci} are thus automatically satisfied, while \eqref{-2ricci} reduces to 
\begin{equation}\label{EFE4pp}
 \Delta H^{(0)} =- \kappa_0  \mathcal{F}^2 , 
\end{equation}
where the RHS (the ``source'' term) depends only on $u$. An example is given by \eqref{F_ex6D}--\eqref{H0_ex6D} with $a=0$.

Finally, {\em conformally flat} solutions (i.e., Weyl type O) can be cast in the form
\be
 W_i=0 , \qquad  H^{(1)} = 0 , \qquad H^{(0)} =-\frac{\kappa_0\mathcal{F}^2}{2(D-2)}\sum_i(x^i)^2 ,
\ee
where a permitted term linear in (or independent of) the $x^i$ in $H^{(0)}$ has been removed by a transformation of the form $x^i\mapsto x^i+h_i(u)$, $r\mapsto r-\dot h_ix^i+g(u)$ (cf. section~24.5 of \cite{Stephanibook}).

The solution \eqref{ppframeIII}, \eqref{localformIII} can be understood as a gravitational and electromagnetic plane-fronted wave propagating in a flat spacetime (recovered for $H^{(1)}=H^{(0)}=W_i=0$). Since $f_{i \dots j}$ in \eqref{localformIII} depends only on $u$, every admissible electromagnetic field $\BF$ is constant over its wave surfaces and hence gives rise to a pure radiation with (transversely) homogeneous energy density $\kappa_0 \mathcal{F}^2/8\pi$. Solutions of type N and O belong to the class of \pp waves \cite{Brinkmann25}, already discussed in a similar context (for particular values of $D$ and $p$) in \cite{Guven87,HorSte90,Horowitz90}.

\begin{remark}[Relation to universal spacetimes and electromagnetic fields]
\label{interpretation}
According to theorem \ref{minuniversal}, Einstein-Maxwell solutions with vanishing higher-order corrections are defined by $VSI$ fields $(\Bg,\BF)$ that satisfy 
\begin{equation}\label{conditions}
\tau_i = 0, \qquad C\indices{_{acde}}C\indices{_{b}^{cde}}=0 .
\end{equation} 
This ensures that, in the limit of a test electromagnetic field (i.e., a ``small'' $\BF$ with negligible backreaction), the solution \eqref{ppframeIII}, \eqref{localformIII} gives rise to a universal electromagnetic field (theorem~1.5 of \cite{HerOrtPra18}) propagating in a Ricci-flat universal spacetime (theorem~1.4 of \cite{HerPraPra14}). However, let us emphasize that the {\em vacuum} solutions of \cite{HerPraPra14} (and \cite{Herviketal15,HerPraPra17}) are more general than the backgrounds allowed by our theorem~\ref{minuniversal}. One reason for this is that we required all higher-order curvature corrections to the Einstein tensor to vanish (and not just be proportional to the metric, as in \cite{HerPraPra14,Herviketal15,HerPraPra17} -- cf. \cite{Coleyetal08} for related comments), which implied that $\Bg$ (as well as $\BF$) is $VSI$. The second reason is that we needed to ensure that also corrections constructed out of $\BF$ vanish, which led to the first of \eqref{conditions} (cf. again the proof of $(i)\Rightarrow (ii)$ in theorem~\ref{minuniversal} for more details). Similarly, also the {\em test} electromagnetic fields on a fixed background obtained in \cite{OrtPra18,HerOrtPra18} are more general than those allowed by our theorem~\ref{minuniversal}, and examples are known in which $\Bg$ and/or $\BF$ are not VSI \cite{OrtPra18,HerOrtPra18}.

In addition, it is worth observing that the metric $\Bg$ defined in \eqref{ppframeIII} can be related to (a subset of the) Ricci-flat universal spacetimes of \cite{HerPraPra14} also by a {\em generalized Kerr-Schild transformation} with a suitable function $\mathcal{H}(u,x)$, under which both \eqref{conditions} are automatically preserved -- in the Kundt coordinates~\eqref{ppframeIII}, this amounts to a change $H^{(0)} \mapsto H^{(0)} + \mathcal{H}$ with $\Delta\mathcal{H}=\kappa_0 \mathcal{F}^2$ ($H^{(0)}$ does not appear in $\ell_{a;b}$ nor in Riemann components of b.w. $0,-1$, which explains why \eqref{conditions} are preserved).

\end{remark}

\begin{remark}[Kerr-Schild form]

Note that \textit{Weyl type N solutions are Kerr-Schild metrics with $\Bell$ (of \eqref{ppframeIII}) being the Kerr-Schild vector, while genuine type III solutions are not} (not even if the Kerr-Schild vector is allowed to be a {\em geodetic} null vector different from $\bl$). The type N part of this statement is manifest using \eqref{WC3pp}. The type III part follows from section~4.2.1 of \cite{OrtPraPra09} (which implies that a spacetime with a Kerr-Schild, Kundt vector is necessarily of Weyl type N, provided the Ricci tensor is N (aligned) or zero) and from Proposition~2 of \cite{OrtPraPra09} (which implies that a spacetime of Weyl type III cannot posses a geodesic Kerr-Schild vector distinct from the (unique) mWAND).

\end{remark}

\begin{remark}[$D=4$ solutions]
As noticed in Remark~\ref{rem_ii}, when $D=4$ the condition $C_{acde}C\indices{_{b}^{cde}} = 0$ can be dropped from theorems~\ref{minuniversal} and \ref {nonminuniversal}, and there are no additional constraints on the spacetime apart from being VSI (and thus Kundt) with a recurrent PND (and satisfying Einstein's equations). Thanks to known results \cite{Stephanibook}, all solutions admitted by theorem~\ref{minuniversal} can thus be reduced to the compact form
\beqn
 & & \d s^2=2\d\zeta\d\bar\zeta-2\d u\left(\d r+W\d\zeta+\bar W\d\bar\zeta+H\d u\right) , \nonumber \\
 & & \bF=\d u\wedge\large[f(u)\d\zeta+\bar f(u)\d\bar\zeta\large] ,
	\label{4D_recurr}
\eeqn 
where
\beqn
 & & W=W(u,\bar\zeta) , \qquad H=\frac{1}{2}(W_{,\bar\zeta}+\bar W_{,\zeta})r+H^{(0)}(u,\zeta,\bar\zeta) , \\
 & & H^{(0)}_{,\zeta\bar\zeta}-\frac{1}{2}\left(W_{,\bar\zeta}^2+\bar W_{,\zeta}^2+WW_{,\bar\zeta\bar\zeta}+\bar W\bar W_{,\zeta\zeta}+W_{,\bar\zeta u}+\bar W_{,\zeta u}\right)=\kappa_0 f\bar f . \label{4D_recurr_Einst}
\eeqn
These spacetimes are in general of Petrov type III. They are of type N iff $W_{,\bar\zeta\bar\zeta}=0$, in which case $W$ can be gauged away \cite{Stephanibook} and one is left with the standard form of electrovac \pp waves $\d s^2=2\d\zeta\d\bar\zeta-2\d u\d r-2H^{(0)}\d u^2$, with $H^{(0)}=\kappa_0f(u)\bar f(u)\zeta\bar\zeta+h(u,\zeta)+\bar h(u,\bar\zeta)$. These solutions were mentioned in a related context in \cite{Coley02}.

Above we discussed the standard case $p=2$. When $p=1$ (or $p=3$ up to duality), the only difference is that the electromagnetic field is given by $\bF=f(u)\d u$, where $f$ is now real, and the RHS of \eqref{4D_recurr_Einst} should be replaced by $\frac{1}{2}\kappa_0 f^2$.
\end{remark}

\begin{remark}[$D=4$ example of Petrov type III]
 In the special case of Einstein gravity coupled to generalized higher-derivative electrodynamics (i.e., $\mathcal{L}_{GC}=0=\mathcal{L}_{int}$), the fact that (ii) implies (i) was already pointed out in \cite{OrtPra16} (but without presenting a proof of this statement). Thanks to theorem~\ref{minuniversal}, a simple four dimensional example of Petrov type III constructed there is also free of corrections in the more general theory~\eqref{lagrangian}. This reads
\beqn
 & & \de s^2 =2\de u\left[\de r+\frac{1}{2}\left(xr-xe^x-2\kappa_{0} e^xc^2(u)\right)\de u\right]+ e^x(\de x^2+e^{2u}\de y^2) , \label{Petrov_EM} \\
 & & \bF=e^{x/2}c(u)\de u\wedge\left(-\cos\frac{ye^u}{2}\de x+e^u\sin\frac{ye^u}{2}\de y\right) .
\eeqn	
It is contained in the more general family \eqref{4D_recurr}, although here it is expressed in slightly different coordinates.
\end{remark}

\section*{Acknowledgments}
%\acknowledgments

We thank Sigbj\o rn Hervik for useful comments. M.O. has been supported by research plan RVO: 67985840 and by the Albert Einstein Center for Gravitation and Astrophysics, Czech Science Foundation GA{\v C}R 14-37086G.

%----------------------------------------------------------------------------------------------------------------------------------
\appendix
%----------------------------------------------------------------------------------------------------------------------------------

\section{Variations of $\mathcal{L}$ evaluated on $VSI$ fields}\label{varapp}
\label{app_variations}
%----------------------------------------------------------------------------------------------------------------------------------
Varying the action \eqref{action} with $\mathcal{L}(I_i,J_j,K_k) \equiv \mathcal{L}_{grav}(I_i) +\mathcal{L}_{elmag}(J_j)+ \mathcal{L}_{int}(K_k)$, one has
\begin{equation}\label{varS}
\delta S = \int \de^{D} x \sqrt{-g}\left( 
-\frac{\mathcal{L}}{2}g_{ab}\delta g^{ab} + \sum_i \frac{\partial  \mathcal{L}_{grav}}{\partial I_i} \delta I_i + \sum_j \frac{\partial \mathcal{L}_{elmag}}{\partial J_j} \delta J_j + \sum_k \frac{\partial  \mathcal{L}_{int}}{\partial K_k} \delta K_k
   \right).
\end{equation}
Let us take a closer look at variation of the individual invariants. Taking e.g. the $n$-th term of the first sum and assuming the boundary terms vanish, integration by parts yields
\begin{equation}
\int \de^{D} x \sqrt{-g} \frac{\partial \mathcal{L}_{grav}}{\partial I_{n}} \delta I_{n} = 
\int \de^{D} x \sqrt{-g} \frac{\partial \mathcal{L}_{grav}}{\partial I_n} \frac{\delta I_n}{\delta g^{ab}} \delta g^{ab} + \left\{\textnormal{terms involving $\nabla^{(k)} \frac{\partial \mathcal{L}_{grav}}{\partial I_n}$}\right\}.
 \label{dL_exp}
\end{equation}
For a $CSI$ metric $\Bg$, the derivatives $\partial\mathcal{L}_{grav}/\partial{I_n}$  are just some constants. Hence, when evaluated on a $CSI$ metric $\Bg$, the bracketed term in \eqref{dL_exp} does not contribute to the resulting variation. A similar argument holds also for the rest of the terms in \eqref{varS}.

Hence, we conclude that,
when evaluated on $CSI$ fields $(\Bg,\BF)$, variations of $\mathcal{L}_{GC}, \mathcal{L}_{EC}$ (recall \eqref{exp_grav} and \eqref{exp_Maxw}\footnote{To avoid possible confusion, let us emphasize that the argument does not really need to assume that \eqref{exp_grav} and \eqref{exp_Maxw} come from a Taylor expansions -- one could alternatively simply {\em define} $\mathcal{L}_{GC}\equiv\mathcal{L}_{grav}-\mathcal{L}_{EH}$ and $\mathcal{L}_{EC}\equiv\mathcal{L}_{elmag} - \mathcal{L}_{M}$ (under the assumption that the Taylor expansions of $\mathcal{L}_{GC}$ and $\mathcal{L}_{EC}$ consist only of terms of higher order, but with no need to take such an expansion).}) and $\mathcal{L}_{int}$ reduce to a linear combination of variations w.r.t. $\delta g^{ab}$ or $\delta F^{ab\dots c}$ of the individual polynomial invariants. 
If, moreover, $(\Bg,\BF)$ are $VSI$, their polynomials invariants $I_k$, $J_k$, $K_i$ 
vanish and so do $\mathcal{L}_{GC}, \mathcal{L}_{EC}$ and $\mathcal{L}_{int}$. Hence, we arrive at the following expressions (evaluated on VSI fields)
\begin{equation}\label{var1}
G_{ab}^{GC}[\Bg] = {16 \pi} \sum_i \frac{\partial \mathcal{L}_{GC}}{\partial I_{i}} (0) \frac{\delta I_i}{\delta g^{ab}} [\Bg], 
\end{equation}
\begin{equation}\label{var2}
T_{ab}^{EC}[\BF] = {-2} \sum_j \frac{\partial \mathcal{L}_{EC}}{\partial J_{j}} (0) \frac{\delta J_j}{\delta g^{ab}} [\BF],
\end{equation}
\begin{equation}\label{variations2}
G_{ab}^{int}[\Bg,\BF] = {16 \pi}  \sum_k \frac{\partial \mathcal{L}_{int}}{\partial K_{k}} (0) \frac{\delta K_k}{\delta g^{ab}} [\Bg,\BF],
\end{equation}
\begin{equation}\label{variations3}
\nabla^a H_{ab \dots c}^{EC}[\BF] =
  {- \frac{8 \pi p}{\kappa_0}} \sum_j \frac{\partial \mathcal{L}_{EC}}{\partial J_{j}} (0) 
\nabla^a \frac{\delta J_j}{\delta F^{ab\dots c}} [\BF] ,
\end{equation}
\begin{equation}\label{var5}
\nabla^a H_{ab \dots c}^{int}[\BF] =
    {- \frac{8 \pi p}{\kappa_0}}  \sum_k \frac{\partial \mathcal{L}_{int}}{\partial K_{k}} (0) 
\nabla^a \frac{\delta K_k}{\delta F^{ab\dots c}} [\BF].
\end{equation}
The above results are used in the proof of theorem~\ref{minuniversal} (see also remark~\ref{CSIvariations}). Similar conclusions (for the metric variations) were obtained in section~4 of \cite{Buchdahl83}. 

For other applications, it may also be useful to note that, since the Lagrangian corrections are of higher-order ($>2$), then necessarily each individual term in the sums in \eqref{var1}--\eqref{var5} is of (the same) higher order. 
Therefore, if one of the scalar invariants $\{I_i,J_j,K_k\}$ is of order 2 (such as $I=R$, $J=F_{a \dots b}F^{a \dots b}$, \dots), then necessarily the partial derivative of $\mathcal{L}_{GC}, \mathcal{L}_{EC}$ and $\mathcal{L}_{int}$ with respect to that invariant vanishes at zero and hence the corresponding term does not contribute to the variation, when evaluated on $VSI$ field. Thus, for example, Einstein's equations for VSI spacetimes are unaffected by higher-order corrections of the form $R^2$ or $R R\indices{_{abcd}} R\indices{^{abcd}}$, but may contain corrections coming, e.g., from $R_{ab}R^{ab}$ (cf. \cite{Buchdahl83} in the special case of 4D \pp waves).

%----------------------------------------------------------------------------------------------------------------------------------
\section{Curvature/electromagnetic rank-2 tensors and $p$-forms}
%----------------------------------------------------------------------------------------------------------------------------------

\subsection{Preliminaries and previous results}

\label{subsec_prelim}

Let us start with some preliminary comments. For the definition of degenerate Kundt spacetimes (needed in the following) we refer the reader to \cite{ColHerPel09a,Coleyetal09} (see also appendix~A of \cite{OrtPra16}), while the definition of balanced and 1-balanced tensors can be found in \cite{Pravdaetal02,Coleyetal04vsi} and \cite{HerPraPra14}, respectively. The GHP (Geroch-Held-Penrose) notation in arbitrary dimension is defined in \cite{Durkeeetal10}.

A null $p$-form is defined by \eqref{nullF}. The traceless part of the Ricci tensor is given by
\be
 S_{ab}\equiv R_{ab}-\frac{R}{D}g_{ab} .
 \label{S}
\ee
In the following, we will mostly consider spacetimes with constant Ricci scalar. It is thus useful to recall
\begin{lemma}[Bianchi identity when $R=$const \cite{GurSisTek14}]
 In a $D$-dimensional spacetime ($D\ge3$) with $R=$const., the following identities hold 
\end{lemma}
\begin{equation}\label{1bianchi}
\nabla^b R_{abcd} = \nabla_d S_{ac} - \nabla_c S_{ad},
\end{equation}
\begin{equation}\label{3bianchi}
\nabla^b C_{abcd} = \frac{D-3}{D-2} ( \nabla_d S_{ac} - \nabla_c S_{ad} ). 
\end{equation}
\begin{proof}
 Just use the contracted Bianchi identity, the definition of the Weyl tensor and \eqref{S}. (For $D=2$ this lemma would be trivial since all the involved quantities vanish identically.)
\end{proof}

Furthermore, we will restrict ourselves to Kundt spacetimes. A Kundt spacetime with constant $R$ is necessarily {\em degenerate} Kundt (cf. Proposition A.2
of \cite{OrtPra16}), for which we have the useful result
\begin{lemma}[Derivatives of 1-balanced tensors in degenerate Kundt spacetimes \cite{HerOrtPra18}]
\label{lemma_1deriv}
 In a degenerate Kundt spacetime, the covariant derivative of a 1-balanced tensor is a balanced 1-tensor.
\end{lemma}

In particular, $VSI$ spacetimes coincide with the Kundt spacetimes of Riemann type III (or more special) \cite{Pravdaetal02,Coleyetal04vsi}, and are therefore a subset of the degenerate Kundt metrics. Recall that 
\begin{lemma}[$\nabla^{(k)}\BR$ in $VSI$ spacetimes \cite{Coleyetal04vsi}]
\label{lemma_VSI}
 In a $VSI$ spacetime, the covariant derivatives $\nabla^{(k)}\BR$ are balanced for any $k \geq 0$.
\end{lemma}

In the rest of this appendix, we will only consider Kundt spacetimes of {\em traceless Ricci type N}, i.e., 
\be
 S_{ab}=\omega'\ell_a\ell_b \qquad (\ell_a \ell^a=0) ,
\ee
where $\omega'$ is a function.

\subsection{New results useful in the proof of theorems~\ref{minuniversal} and \ref {nonminuniversal}}

With a mild assumption on $R$ (i.e., not necessarily constant) one can prove
\begin{lemma}\label{Sbalanced}
Let $\Bg$ be a traceless Ricci type N Kundt metric with $\tho' R = 0$ in a frame adapted to $\bl$. Then $\nabla^{(k)}\BS$ is 1-balanced for
any $k \geq 0$.
\end{lemma}
\begin{proof}
Under the assumptions, the contracted Bianchi identity implies $\tho \omega^\prime = 0$ (cf. the primed
version of (2.50, \cite{Durkeeetal10}), or (2.35, \cite{KucPra16})). Therefore, the definition of 1-balanced tensors, together
with lemma~\ref{lemma_1deriv}, implies that the traceless part of the Ricci tensor and its covariant
derivatives of arbitrary order are 1-balanced.
\end{proof}

\begin{lemma}\label{formlemma}
Let $\Bg$ be a Weyl type III, Ricci type N Kundt metric. There is no non-vanishing $p$-form constructed from $\BR$ and its covariant derivatives for $p\ge0$. 
\end{lemma}
\begin{proof}
Before starting, let us note that, under the assumptions, $\nabla^{(k)}\BR$ is balanced for any $k\ge 0$ (lemma~\ref{lemma_VSI}), while $\nabla^{(l)}\BS$ is 1-balanced  for any $l\ge 0$ (lemma~\ref{Sbalanced}). This will be implicitly used in the following.

First, in the case of $p=0$, such form would be a curvature scalar, which here vanishes since $\Bg$ is $VSI$. The dual case $p=D$ is treated analogously.
Let us thus discuss the $0<p<D$ case.
If there was a non-vanishing $p$-form $\BH[\BR, \nabla \BR, \dots]$ (with boost order at least $(-1)$), it would be necessarily at most linear in $\nabla^{(k)}\BC$, $k \geq 0$, and from the Ricci identity, it follows that covariant derivatives in $\nabla^{(k)} \BC$ effectively commute (i.e., up to terms of b.w. -2). 

Consider the $p\leq 2$ case. Since $\BC$ is traceless, there are necessarily at least two contractions of a derivative index with a Weyl tensor index within $\nabla^{(k)}\BC$. After commuting derivatives and employing \eqref{3bianchi}, we observe that such a contraction is of boost order $(-2)$ and $\BH$ vanishes. 

Now, let us discuss the $p>2$ case. Assume that we are able to obtain some non-vanishing rank-$p$ contraction of $\nabla^{(k)}\BC$. In order to produce a $p$-form, it has to be antisymmetrized over all remaining $p\geq 3$ indices. But now the Bianchi identities come into play
\begin{equation}\label{bianchi}
R_{a[bcd]} = 0, \qquad R_{ab[cd;e]}=0. 
\end{equation}
Since covariant derivatives in $\nabla^{(k)}\BC$ effectively commute, the antisymmetrization has to be performed over at most one derivative index and at least two Weyl tensor indices. But from \eqref{bianchi}, it follows that (after shuffling the derivatives if needed) the result is zero anyway. 
\end{proof}

\begin{lemma}\label{rank2FC}
Let $\Bg$ be a Weyl type III, Ricci type N Kundt metric and $\BF$ be an aligned null $p$-form. There is no non-vanishing symmetric rank-2 contraction of $\nabla^{(k)}\BC \otimes \BF$ and $\nabla^{(k)}\BR \otimes \BF$ for $k\geq 0$.
\end{lemma}
\begin{proof}
Both $\BF$ and $\nabla^{(k)}\BC$ are of boost order $(-1)$. 
Due to skew-symmetry of $\BF$, at most one of its indices can be left uncontracted, while each of the rest of the indices of $\BF$ has to be contracted with 
some index of $\nabla^{(k)}\BC$. 
Moreover, covariant derivatives of $\BC$ again effectively commute. 

If $p>3$, this necessarily yields antisymmetrization of $\nabla^{(k)}\BC$ over at least 3 indices, which is zero due to Bianchi identities and effective commutativity of covariant derivatives of $\BC$, as we saw in the proof of lemma \ref{formlemma}. 

For $p \leq 3$, there is either one index of $\BF$ left uncontracted (and hence there is necessarily a contraction of indices within $\nabla^{(k)}\BC$) or each of the indices of $\BF$ is contracted with some index of $\nabla^{(k)}\BC$.  
However, any contraction within $\nabla^{(k)}\BC$ will eventually (after commuting the derivatives) vanish, 
since $\nabla^aC_{abcd}$, and consequently also $\Box C_{abcd}$, are (recall \eqref{3bianchi}) of boost order $(-2)$. The first case thus cannot produce any non-vanishing result. 
In the second case, one can easily verify that the corresponding contraction vanishes again due to skew-symmetry of $\BF$ and Bianchi identities \eqref{bianchi}. 

That the same result holds also for $\nabla^{(k)}\BR \otimes \BF$ follows from the Ricci tensor being of type N (and a trivial b.w. counting).
\end{proof}

\begin{lemma}
\label{rank2F}
Let $\Bg$ be a Weyl type III, Ricci type N Kundt spacetime and $\BF$ an aligned null Maxwell $p$-form. If $\nabla \BF$ is $1$-balanced, then all non-vanishing symmetric rank-2 tensors constructed from $\BF$ and its covariant derivatives are of second order.
\end{lemma}
\begin{proof}
 By simple b.w. counting, terms cubic in $\BF$ and quadratic in $\nabla^{(k)}\BF$ ($k>0$) cannot contribute (and similarly for higher powers), while terms quadratic in $\bF$ are obviously of second order.  Terms linear in $\BF$ cannot contribute because of its total antisymmetry. It remains to be shown that also terms linear in $\nabla^{(k)}\BF$ do not contribute. Let us first discuss the case $1<p<D-1$. By the symmetry of the indices there must be at least one contraction of an index of $\BF$ with one derivative index. The idea is thus to use commutators of covariant derivatives and the Maxwell equations to show that all such terms vanish. By the Ricci identity and 1-balancedness of $\nabla^{(k)}\BF$, commutators $[\nabla,\nabla]\nabla^{(k)}\BF$ with $k>0$ are  
(recalling also lemmas~\ref{lemma_1deriv} and \ref{lemma_VSI}) of b.w. $-3$ and therefore do not contribute. The only non-trivial commutator is thus $[\nabla,\nabla]\BF$ (and its derivatives). This gives terms which are contractions of $\nabla^{(l)}\BC \otimes \BF$ for $l\geq 0$ (up to terms of b.w. $-3$), which indeed do not contribute thanks to lemma~\ref{rank2FC}. This completes the proof for $1<p<D-1$. When $p=1$ (or, by duality, $p=D-1$), symmetric 2-tensors can be constructed out of $\nabla^{(k)}\BF$ even without contracting an index of $\BF$ with one derivative index. For $k=1$ this gives the term $\nabla_{(a}F_{b)}$, which is of order 2. For $k>1$ (which requires $k\ge3$) there is at least a contraction between two derivative indices. Similarly as above, derivatives in such terms can thus be shuffled to obtain $\nabla^{k-2} \Box \BF$, which vanishes thanks to Maxwell's equations and the Weitzenb\"ock identity (cf. eq.~(12) of \cite{HerOrtPra18}).\footnote{Throughout the proof we did not discuss explicitly terms constructed using the dual $(D-p)$-form $\star \BF$. However, all the steps still apply, since $\star \BF$ is automatically aligned with $\bF$ and inherits from it all the essential properties.} 
\end{proof}

\begin{remark}[Terms of second order]
For completeness, let us observe that, under the assumptions of lemma~\ref{rank2F}, terms quadratic in $\bF$ generically reduce to (constant multiples of) $F\indices{_{a c \dots d}} F \indices{_b^{c \dots d}}={\mathcal F}^2\bl\bl$. In the special case $n=2p$ with $p$ odd, another possible term is $F_{ac\ldots d}\star\!F_b^{\ c\ldots d}$, which is in general non-zero and different from $F_{ac\ldots d}F_b^{\ c\ldots d}$ (but still $\propto\bl\bl$). For $n=2p$ with $p$ even, instead, such a term vanishes identically thanks to the identity $F_{ac\ldots d}\star\!F_b^{\ c\ldots d}[1+(-1)^{p^2}]=\frac{1}{p}(F_{cd\ldots e}\star\!F^{cd\ldots e})g_{ab}$ (since $\BF$ is VSI and thus $F_{cd\ldots e}\star\!F^{cd\ldots e}=0$). We further note that, for $p=1$, the term $\nabla_{(a}F_{b)}$ is also proportional to $\ell_a\ell_b$ (just by b.w. counting), but in general different from $F_a F_b$. 
\end{remark}

%----------------------------------------------------------------------------------------------------------------------------------

%----------------------------------------------------------------------------------------------------------------------------------
\section{Rank-2 curvature tensors in recurrent spacetimes of Weyl type III and traceless Ricci type N with $C_{acde}C\indices{_{b}^{cde}} = 0$}
%----------------------------------------------------------------------------------------------------------------------------------

In \cite{Gursesetal13} (cf. also \cite{HerPraPra14}), it was shown that if $\Bg$ is a Weyl type N and traceless Ricci type N Kundt metric with a constant Ricci scalar (in which case $\Bg$ is necessarily  $CSI$, see Corollary~A.5 of \cite{KucPra17} and Remark~A.9 of \cite{HerOrtPra18}), then any symmetric rank-2 tensor constructed from the Riemann tensor and its covariant derivatives takes the form $T_{ab} = \lambda g_{ab} + \sum_{n=0}^N a_n \Box^n S_{ab}$, 
where $\lambda$ and $a_n$ are some constants and $N\in \mathbb{N}$. It has been recently shown that the assertion can be extended also to Weyl type III, provided the Weyl tensor satisfies certain conditions \cite{Kuchynkaetal18}. A special subcase of Proposition~6 of \cite{Kuchynkaetal18} (cf. also~(14) therein), useful for our purposes, can be formulated as
\begin{theorem}[On symmetric rank-2 tensors \cite{Kuchynkaetal18}]
\label{Tform}
Let $\Bg$ be a Weyl type III and traceless Ricci type N metric such that: (i) the mWAND is recurrent; (ii) $C_{acde}C\indices{_{b}^{cde}} = 0$. Then any symmetric rank-2 tensor constructed from the Riemann tensor and its covariants derivatives of arbitrary order takes the form 
\begin{equation}\label{Tformeqn}
T_{ab} = \sum_{n=0}^N a_n \Box^n S_{ab}.
\end{equation}  
\end{theorem}
For self-containedness, let us present a proof tailored to this special case.
\begin{proof}
Before starting we observe that the Weyl and Ricci tensors are necessarily aligned thanks to proposition~3.1 of \cite{KucPra16}. Then, the line of the proof will be similar to that of \cite{Gursesetal13}. However, in contrast with the Weyl type N case, also various contractions of the Weyl tensor and its covariant derivatives can in principle contribute to $\BT$ \cite{HerPraPra14}. But under the additional conditions, we will show that any of these actually vanishes, so that one is left with $\BT$ of the form \eqref{Tformeqn}. 

First, $\tau_i = 0$ implies that the Ricci scalar vanishes (cf., e.g., Remark~A.9 of \cite{HerOrtPra18}). Also,  $\BC$ is balanced and $\BS$ is 1-balanced (lemma \ref{Sbalanced}), and hence the only possible contributions to $\BT$ come from contractions of $\nabla^{(k)}\BR$ and of $\nabla^{(k)} \BC \otimes \nabla^{(l)}\BC$ with $k,l \geq 0$. In particular, $\BT$ is traceless. 

Now, let us focus on contractions of $\nabla^{(k)}\BR$ (clearly, $k$ has to be even). 
From the Ricci identity, it is obvious that any change in the order of covariant derivatives in $\nabla^{(k)}\BR$ produces only terms of type $\nabla^{(m} \BC \otimes \nabla^{(l)}\BC$. 
Following the procedure sketched in \cite{Gursesetal13} with use of \eqref{1bianchi} and $\nabla^b S_{ab} = 0$ (since $\BS$ is 1-balanced),
any contraction of $\nabla^{(k)}\BR$ can be cast in the form linear in $\Box^{k/2} \BS$ plus terms quadratic in the Weyl tensor and its covariant derivatives. 

At this moment, to finish the proof of the assertion, it is sufficient to show that all rank-2 contractions of $\nabla^{(k)} \BC \otimes \nabla^{(l)}\BC$ vanish. This can be done employing \eqref{3bianchi}
and following step by step the procedure of section 5.1 in \cite{HerPraPra14}. In this manner, one obtains an extension of proposition 5.6 of \cite{HerPraPra14} to the Ricci type N case, which completes the proof. 
\end{proof}

%----------------------------------------------------------------------------------------------------------------------------------
\section{On covariant derivatives of null $p$-forms in Kundt spacetimes}  
%----------------------------------------------------------------------------------------------------------------------------------
\label{app_nullF}
In this section, we will provide some useful results on null $p$-forms and their covariant derivatives. A $p$-form $\bF$ is {\em null} iff it can be written as \cite{OrtPra16}
\be
	\bF=\bl\wedge\bff , \qquad \ell_a\ell^a=0=f_{a\ldots b}\ell^a ,
	\label{nullF}
\ee
where $\bff$ is a $(p-1)$-form. In other words, $\bF$ possesses only components of b.w. $-1$ \cite{Durkeeetal10}. Obviously this is possible only for $1\le p\le D-1$.

\begin{remark}[Maxwell's equations]
If one assumes that $\bl$ in \eqref{nullF} is Kundt, in a null frame adapted to $\bl$ the GHP Maxwell equations reduce to \cite{Durkeeetal10} (cf. also eqs.~(2.16)--(2.18) of \cite{KucPra17} -- $\phip_{i j \dots k}$ is denoted $\varphi'_{i j \dots k}$ in \cite{Durkeeetal10,KucPra17})
\beqn
 & & \dho_i \phip_{ij \dots k} =\tau_i \phip_{i j \dots k} , \label{Max1} \\
 & & \dho_{[i}\phip_{j \dots k]} =\tau_{[i}\phip_{j \dots k]} , \label{Max2} \\
 & & \tho\phip_{i\dots j}=0 . \label{Max3}
\eeqn
\end{remark}

If $\bl$ in \eqref{nullF} is Kundt, $\nabla \BF$ possesses generically non-zero components of b.w. $0,-1,-2$. More precisely, defining the standard directional derivatives $D \equiv \ell^a \nabla_a$, $\T\equiv n^a \nabla_a$, $\delta_i \equiv m^{(i)a} \nabla_a$,
\begin{lemma}\label{dFcomponents}
Let $\bF=\bl\wedge\bff$ be a null $p$-form and $\bl$ a Kundt vector field. Then, in a null frame adapted to $\bl$ (i.e., with $\Bm_{(0)} = \Bell$ but otherwise arbitrary) 
\begin{enumerate}
\item $D\BF$ is null with frame components $(DF)_{1i \dots j} = \tho\phip_{i\dots j}$;
\item\label{dFcomponents_ii} $\delta_i \BF$ is null with frame components $(\delta_i F)_{1j \dots k} = \dho_i \phip_{j \dots k}$; 
\item $\T \BF$ is of type II with frame components $(\T F)_{01j \dots k} =\tau_i \phip_{i j \dots k}$, $(\T F)_{ij \dots k} =p\tau_{[i}\phip_{j \dots k]}$ and  $(\T F)_{1i \dots j} = \tho^\prime \phip_{i \dots j}$.
\end{enumerate}
\end{lemma}
\begin{proof}
The result follows by a direct calculation of the various frame components of $\nabla \BF$.
\end{proof}

\begin{lemma}
 \label{lemma_DFDF}
 Let $\bF=\bl\wedge\bff$ be a null $p$-form and $\bl$ a Kundt vector field. If $\tho\phip_{i\dots j}=0$, then 
 \be
  \nabla_c F_{ad \dots e} \nabla^c F\indices{_b^{d \dots e}} =(\dho_i \phip_{j \dots k})(\dho_i \phip_{j \dots k})\ell_a\ell_b . 
	\label{DFDF}
 \ee
\end{lemma}
\begin{proof}
 Thanks to lemma~\ref{dFcomponents}, we know that $\nabla \BF$  has only components of negative b.w.. The contraction over $c$ in \eqref{DFDF} further ensures that only the components \eqref{dFcomponents_ii} of lemma~\ref{dFcomponents} contribute, and the result thus follows.
\end{proof}
\begin{remark}
 \label{rem_DFDF}
 The assumption $\tho\phip_{i\dots j}=0$ in lemma~\ref{lemma_DFDF} is satisfied identically if $\bF$ is a Maxwell field (eq.~\eqref{Max3}).
\end{remark}

The special case when $\nabla \BF$ has only components of b.w. $-2$ can be characterized as follows.
\begin{lemma}\label{FtypeNchar}
Let $\BF$ be a non-vanishing null $p$-form. Then, $\nabla \BF$ is of type N (necessarily aligned) iff $\bl$ is Kundt and the scalars $\tho\phip_{i\dots j}$,  $\dho_i \phip_{j \dots k}$, $\tau_i $ vanish.
\end{lemma}
\begin{proof}
The type N condition means that  $\nabla \BF$ possesses only components of b.w. $-2$ (which is possible only if $\nabla \BF$ is aligned with $\BF$, since $\BF$ is null). Proposition~C.1 of \cite{OrtPra16} implies that $\bl$ is Kundt and $\tho\phip_{i\dots j}=0$. Using lemma~\ref{dFcomponents} further gives $\dho_i \phip_{j \dots k}=0$ and $\tau_i \phip_{i j \dots k} = 0 =\tau_{[i}\phip_{j \dots k]}$. Since $\phip_{i \dots k}\neq0$, the last two equations imply $\tau_i = 0$. (For $p=1$ the equation $\tau_i \phip_{i j \dots k} = 0$ does not appear, but the conclusion is unchanged.) The other direction of the lemma can be proven by just reversing the above steps.
\end{proof}

\begin{remark}
 The fact that $\bl$ is Kundt and $\tau_i=0$ is equivalent to saying that $\bl$ is {\em recurrent}. In addition, note that, in particular, a null $\bF$ with $\nabla \BF$ of type N satisfies Maxwell's equations identically (cf. eqs.~\eqref{Max1}--\eqref{Max3}).
\end{remark}

\begin{lemma}\label{Maxwelltau}
Let $\Bg$ be a spacetime of Weyl type III and  $\BF$ an aligned null $p$-form $\BF$ such that $\nabla \BF$ is of type N. If $(\Bg,\BF)$ is a solution of the Einstein-Maxwell equations, necessarily $\Lambda=0$ and both $(\Bg,\BF)$ are $VSI$.  
\end{lemma}

\begin{proof}
 From lemma~\ref{FtypeNchar} we have that $\bl$ is recurrent (and thus Kundt). The Einstein equations imply that the traceless Ricci type is N (and that $R$ is proportional to $\Lambda$), therefore the spacetime is Kundt degenerate. A non-vanishing Ricci scalar would require $\tau_i\neq0$ (cf., e.g., Remark~A.9 of \cite{HerOrtPra18}), therefore $\Lambda=0$. The Ricci type is thus N and the VSI property of $\Bg$ then follows immediately from theorem~1 of \cite{Coleyetal04vsi}. Finally, the VSI property of $\BF$ follows from theorem~1.5 of \cite{OrtPra16} (since $\tho\phip_{i\dots j}=0$ by lemma~\ref{FtypeNchar}).
\end{proof}

\begin{lemma}\label{typeNchar}
Let $\BF$ be a non-vanishing null Maxwell field. Then $\nabla \BF$ is of type N iff $\dho_i \phip_{j \dots k} = 0$  and $\bl$ is Kundt.
\end{lemma}
\begin{proof}
Maxwell's equations \eqref{Max1}, \eqref{Max2} guarantee that, if a non-vanishing null solution $\BF$ in a Kundt spacetime satisfies $\dho_i \phip_{j \dots k} = 0$, then automatically also $\tau_i = 0$. 
By lemma~\ref{FtypeNchar}, the ``if'' assertion follows. The same lemma ensures that also the ``only if'' direction holds.
\end{proof}

\begin{lemma}\label{1balcriteria}
Let $\BF$ be a non-vanishing null Maxwell field in an aligned Weyl and traceless Ricci type III spacetime. Then $\nabla \BF$ is 1-balanced iff $\dho_i \phip_{j \dots k} = 0$ and $\bl$ is Kundt. 
\end{lemma}
\begin{proof}
Thanks to lemma~\ref{typeNchar}, we know that $\dho_i \phip_{j \dots k} = 0$ and $\bl$ Kundt are necessary conditions for 1-balancedness of $\nabla \BF$. To show that these conditions are also sufficient, it remains to verify (lemma \ref{dFcomponents}) that $D \tho^\prime \phip_{i \dots j} = 0$ in an affinely parametrized, parallely propagated frame. Since $D \phip_{i \dots j}=D L_{11}=D M^i_{j1}=0$ and $[\T,D]=L_{11}D$ (thanks to the Kundt and curvature assumptions, cf., e.g., appendix~A.1 of \cite{OrtPra16}), the assertion follows.  
\end{proof}

\begin{remark}
 \label{rem_DFDF_2}
 Thanks to lemma~\ref{lemma_DFDF} and remark~\ref{rem_DFDF}, the condition $\dho_i \phip_{j \dots k} = 0$ in theorem~\ref{1balcriteria} can equivalently be written in a covariant form as 
\be
 \nabla_c F_{ad \dots e} \nabla^c F\indices{_b^{d \dots e}} =0. 
\ee
Let us emphasize that for a null Maxwell field aligned with a Kundt direction, this condition implies that $\bl$ is recurrent (as observed in the proof of lemma~\ref{typeNchar}).
\end{remark}

\begin{remark}
 \label{rem_DFDF_vsi}
 In the special case of a null Maxwell field in a VSI spacetime, in the canonical coordinates and frame of \cite{Coleyetal06}, Maxwell's equations imply $\phip_{j \dots k,r} = 0$, while the condition $\dho_i \phip_{j \dots k} = 0$ of theorem~\ref{1balcriteria} simplify reads $\phip_{j \dots k,l} = 0$. Therefore, $\phip_{j \dots k}$ is only a function of $u$. All Maxwell's equations \eqref{Max1}--\eqref{Max3} are then satisfied identically (cf. also \cite{OrtPra16}).
\end{remark}

%
%%----------------------------------------------------------------------------------------------------------------------------------
%%\nocite{*}
%\bibliographystyle{unsrt}
%\renewcommand\refname{References}
%\bibliography{bibl}  
%%\bibliography{bibl}  
%%----------------------------------------------------------------------------------------------------------------------------------

%----------------------------------------------------------------------------------------------------------------------------------
\end{document}